\documentclass[10pt,twoside,twocolumn,english,aps,manuscript,nofootinbib, nobibnotes, secnumarabic]{revtex4}
\usepackage[T1]{fontenc}
\usepackage[latin9]{inputenc}
\usepackage{xcolor}
\usepackage{calc}
\usepackage{units}
\usepackage{amsthm}
\usepackage{amsmath}
\usepackage[all]{xy}
\PassOptionsToPackage{normalem}{ulem}
\usepackage{ulem}

\makeatletter

\providecommand{\tabularnewline}{\\}
\providecolor{lyxadded}{rgb}{0,0,1}
\providecolor{lyxdeleted}{rgb}{1,0,0}

\@ifundefined{textcolor}{}
{%
 \definecolor{BLACK}{gray}{0}
 \definecolor{WHITE}{gray}{1}
 \definecolor{RED}{rgb}{1,0,0}
 \definecolor{GREEN}{rgb}{0,1,0}
 \definecolor{BLUE}{rgb}{0,0,1}
 \definecolor{CYAN}{cmyk}{1,0,0,0}
 \definecolor{MAGENTA}{cmyk}{0,1,0,0}
 \definecolor{YELLOW}{cmyk}{0,0,1,0}
}
\theoremstyle{plain}
\newtheorem{thm}{\protect\theoremname}
  \theoremstyle{plain}
  \newtheorem{lem}[thm]{\protect\lemmaname}

\usepackage[all]{xy}

\usepackage{multirow}

\def\frontmatter@abstractheading{}

\makeatother

\usepackage{babel}
  \providecommand{\lemmaname}{Lemma}
\providecommand{\theoremname}{Theorem}

\begin{document}

\title{Generalizing Bell-type and Leggett-Garg-type Inequalities to Systems
with Signaling }

\author{Ehtibar N. Dzhafarov}

\affiliation{Purdue University}

\author{Janne V. Kujala}

\affiliation{University of Jyväskylä}
\begin{abstract}
\mbox{}

Contextuality means non-existence of a joint distribution for random
variables recorded under mutually incompatible conditions, subject
to certain constraints imposed on how the identity of these variables
may change across these conditions. In simple quantum systems contextuality
is indicated by violations of Bell-type or Leggett-Garg-type inequalities.
These inequalities, however, are predicated on the assumption of no-signaling,
defined as invariance of the distributions of measurement results
with respect to other (e.g., earlier in time) measurements' settings.
Signaling makes the inequalities inapplicable: a non-signaling system
with any degree of contextuality, however high, loses any relation
to this concept as soon as it exhibits any degree of signaling, however
small. This is unsatisfactory. We describe a principled way of defining
and measuring contextuality in arbitrary systems with random outputs,
whether signaling{\normalsize{} is absent or present. }{\normalsize \par}

\mbox{}

\textsc{Keywords:} Bell/CHSH inequalities; contextuality; EPR/Bohm
paradigm; Leggett-Garg inequalities; signaling.

\markboth{Dzhafarov and Kujala}{Contextuality on top of Signaling}
\end{abstract}
\maketitle

\section{Introduction}

Contextuality can be defined in purely probabilistic terms, for abstract
systems with random outputs recorded under different (mutually incompatible)
conditions \cite{1Specker,2Larsson,3Khrennikov,4Cabello,5DK2013LNCS,6DK2014PLOS1,7DK2014arxiv}.
Consider, e.g., $\left(X_{1},Y_{1},Z_{1},\ldots\right)$ recorded
under condition $c_{1}$, $\left(X_{2},Y_{2},Z_{2},\ldots\right)$
recorded under condition $c_{2}$, etc. The notion of contextuality
involves a hypothesis that certain random variables preserve their
identity across some of the different conditions: e.g., that $X_{1}=X_{2}$.
The system exhibits no contextuality (with respect to this hypothesis)
if all the random variables $\left(X_{i},Y_{i},Z_{i},\ldots\right)$
across different values of $i$ can be viewed as jointly distributed
with $X_{1}$ and $X_{2}$ being always equal to each other. In the
Kolmogorovian probability theory, being jointly distributed is equivalent
to the random outputs being (measurable) functions of one and the
same (``hidden'') random variable $\lambda$ \cite{8NHMT}:
\begin{equation}
X_{i}=x_{i}\left(\lambda\right),Y_{i}=y_{i}\left(\lambda\right),Z_{i}=z_{i}\left(\lambda\right),\ldots.
\end{equation}
The constraint $X_{1}=X_{2}$ means
\begin{equation}
\Pr\left[X_{1}\not=X_{2}\right]=\Pr\left[\lambda:x_{1}\left(\lambda\right)\not=x_{2}\left(\lambda\right)\right]=0.
\end{equation}
As a well-known example, in the simplest Alice-Bob EPR/Bohm paradigm
\cite{9CHSH,10CH}, the four mutually incompatible conditions $\left(\alpha_{i},\beta_{j}\right)$
are formed by Alice's settings $\alpha_{1}$ or $\alpha_{2}$ combined
with Bob's settings $\beta_{1}$ or $\beta_{2}$. Under each condition
$\left(\alpha_{i},\beta_{j}\right)$, Alice and Bob record spins represented
by binary ($\pm1$) random variables $A_{ij}$ and $B_{ij}$, respectively.
We will refer to a system with this input-output relation as a Bell-system.
It involves eight random variables, with the joint distribution being
known for each pair $\left(A_{ij},B_{ij}\right)$ but not across different
pairs. The identity hypothesis here is that $A_{i1}=A_{i2}$ for $i=1,2$,
and $B_{1j}=B_{2j}$ for $j=1,2$. Stated rigorously, if one can impose
a joint distribution on all eight random variables consistent with
the known distributions of $\left(A_{ij},B_{ij}\right)$ and constrained
by the requirement
\begin{equation}
\Pr\left[A_{i1}\not=A_{i2}\right]=\Pr\left[B_{1j}\not=B_{2j}\right]=0,\; i,j\in\left\{ 1,2\right\} ,\label{eq:identity connections B}
\end{equation}
then the Bell system exhibits no contextuality.

Similarly, in the simplest Leggett-Garg paradigm \cite{11Leggett},
there are three mutually exclusive conditions $\left(t_{1},t_{2}\right)$,
$\left(t_{1},t_{3}\right)$, and $\left(t_{2},t_{3}\right)$, formed
by three fixed time moments $t_{1}<t_{2}<t_{3}$. The two binary ($\pm1$)
random outputs jointly recorded at moments $t_{i}<t_{j}$ can be denoted
$Q_{ij}$ and $Q_{ji}$, respectively. We will refer to a system with
this input-output relation as an LG-system. It involves six random
variables, with the joint distribution known for each pair $\left(Q_{ij},Q_{ji}\right)$
but not across different pairs. The identity hypothesis here is that
$Q_{12}=Q_{13}$, $Q_{21}=Q_{23}$, and $Q_{31}=Q_{32}$. The LG-system
exhibits no contextuality if one can impose a joint distribution on
all six random variables consistent with the known distributions of
the pairs $\left(Q_{ij},Q_{ji}\right)$ and subject to
\begin{equation}
\Pr\left[Q_{12}\not=Q_{13}\right]=\Pr\left[Q_{21}\not=Q_{23}\right]=\Pr\left[Q_{31}\not=Q_{32}\right]=0.\label{eq:identity connections LG}
\end{equation}

The issue we take on in this paper is related to the fact that non-contextuality
defined as above implies the condition known as marginal selectivity
\cite{8NHMT,12DK2013ProcAMS} or no-signaling \cite{13Cereceda,14Masanes}:
obviously, any set of random variables whose identity is preserved
across different conditions preserves its distribution across these
conditions. For the Bell-systems, no-signaling means, using $\left\langle .\right\rangle $
for expected value,
\begin{equation}
\left\langle A_{i1}\right\rangle =\left\langle A_{i2}\right\rangle ,\;\left\langle B_{1j}\right\rangle =\left\langle B_{2j}\right\rangle ,\; i,j\in\left\{ 1,2\right\} ,\label{eq:marginal selectivity B}
\end{equation}
while for the LG-systems it means
\begin{equation}
\left\langle Q_{ij}\right\rangle =\left\langle Q_{ij'}\right\rangle ,\; i,j,j'\in\left\{ 1,2,3\right\} ,\; i\not=j,\; i\not=j'.\label{eq:marginal selectivity LG}
\end{equation}
The necessary and sufficient condition for non-contextuality in the
two types of systems are obtained as conjunctions of the no-signaling
requirements just given with certain inequalities involving jointly
distributed pairs: for the Bell-systems it is the conjunction of (\ref{eq:marginal selectivity B})
with the CHSH inequality \cite{15Fine}

\begin{equation}
\max_{i,j\in\left\{ 1,2\right\} }\left|\begin{array}{l}
\left\langle A_{11}B_{11}\right\rangle +\left\langle A_{12}B_{12}\right\rangle \\
+\left\langle A_{21}B_{21}\right\rangle +\left\langle A_{22}B_{22}\right\rangle -2\left\langle A_{ij}B_{ij}\right\rangle 
\end{array}\right|\leq2,\label{eq:CHSH}
\end{equation}
while for the LG-systems it is the conjunction of (\ref{eq:marginal selectivity LG})
with the Leggett-Garg-Suppes-Zanotti (LGSZ) inequality \cite{11Leggett,16Suppes,17note1}

\begin{equation}
\begin{array}{l}
-1\le\left\langle Q_{12}Q_{21}\right\rangle +\left\langle Q_{13}Q_{31}\right\rangle +\left\langle Q_{23}Q_{32}\right\rangle \\
\leq1+2\min\left\{ \left\langle Q_{12}Q_{21}\right\rangle ,\left\langle Q_{13}Q_{31}\right\rangle ,\left\langle Q_{23}Q_{32}\right\rangle \right\} .
\end{array}\label{eq:LGSZ}
\end{equation}
The inequalities are logically independent of the corresponding no-signaling
conditions: one can construct examples of systems with all four combinations
of truth values for (\ref{eq:marginal selectivity B}) and (\ref{eq:CHSH}),
or for (\ref{eq:marginal selectivity LG}) and (\ref{eq:LGSZ}) \cite{18DK2013Topics}.

Logically, then, we should consider a Bell-system exhibiting contextuality
if either CHSH inequalities (\ref{eq:CHSH}) are violated or no-signaling
condition (\ref{eq:marginal selectivity B}) is violated (or both);
and analogously for the LG-systems. However, to posit that any instance
of signaling constitutes contextuality amounts to unreasonably expanding
the meaning of contextuality, and it contradicts the common usage.
If changes in Bob's setting somehow change the distribution of spins
recorded by Alice under a fixed setting (assuming the two are separated
by a time-like interval), the natural language to use is that of direct
cross-influences rather than contextuality. But it is equally unsatisfactory
to declare (non-)contextuality undefined whenever signaling is present.
Consider, e.g., a Bell system with
\begin{equation}
\begin{array}{c}
\left\langle A_{11}B_{11}\right\rangle =\left\langle A_{12}B_{12}\right\rangle =\left\langle A_{21}B_{21}\right\rangle =-\left\langle A_{22}B_{22}\right\rangle =\delta,\\
\left\langle A_{11}\right\rangle =\left\langle B_{11}\right\rangle =\left\langle A_{12}\right\rangle =\left\langle B_{12}\right\rangle =\left\langle A_{21}\right\rangle =\left\langle B_{21}\right\rangle =0,\\
\left\langle A_{22}\right\rangle =-\left\langle B_{22}\right\rangle =\varepsilon.
\end{array}\label{eq:PR+signaling}
\end{equation}
It satisfies the no-signaling condition (\ref{eq:marginal selectivity B})
if and only if $\varepsilon=0$. In this case, for any $\delta>\nicefrac{1}{2}$,
it violates CHSH inequalities (\ref{eq:CHSH}), indicating thereby
contextuality. If the degree of contextuality is measured as proportional
to the excess of the left-hand side of (\ref{eq:CHSH}) over $2$,
the maximum contextuality allowed by quantum mechanics \cite{19Cirel'son}
is achieved at $\delta=\nicefrac{1}{\sqrt{2}}$, whereas $\delta=1$
represents a Bell-system with maximum contextuality algebraically
possible \cite{20Popescu}. But as soon as $\varepsilon$ differs
from zero, however slightly, contextuality changes from a very high
(even highest possible) level to being undefined. Among other things,
this creates difficulties for statistical analysis of contextuality,
where one can never establish with certainty that equalities (\ref{eq:marginal selectivity B})
and (\ref{eq:marginal selectivity LG}) hold precisely. 

In this paper we propose a new definition and new measure of contextuality
that overcome this difficulty: even in the presence of direct cross-influences
(say, from Bob's setting to Alice's measurements and vice versa) one
can identify and compute the degree of contextual influences ``on
top of'' the direct cross-influences.

\section{Criterion for (Non)Contextuality}

The main idea is this: contextuality is present if random variables
recorded under different conditions cannot be presented as a single
system of jointly distributed random variables, provided their identity
across different conditions changes as little as it is possible in
view of the observed differences between marginal distributions (i.e.,
in view of signaling). 

For a Bell-system, we consider the vector of probabilities \cite{note2}
\begin{equation}
C=\left(\begin{array}{l}
\Pr\left[A_{11}\not=A_{12}\right],\Pr\left[A_{21}\not=A_{22}\right],\\
\Pr\left[B_{11}\not=B_{21}\right],\Pr\left[B_{12}\not=B_{22}\right]
\end{array}\right)
\end{equation}
and find the minimum possible values of these probabilities allowed
by the system's marginal expectations 
\begin{equation}
\left(\begin{array}{l}
\left\langle A_{11}\right\rangle ,\left\langle A_{12}\right\rangle ,\left\langle A_{21}\right\rangle ,\left\langle A_{22}\right\rangle ,\\
\left\langle B_{11}\right\rangle ,\left\langle B_{21}\right\rangle ,\left\langle B_{12}\right\rangle ,\left\langle B_{22}\right\rangle 
\end{array}\right).\label{eq:state signaling B}
\end{equation}
Denote this vector $C$ by $C_{0}$. It is specified as follows.
\begin{lem}
\label{lem:C0 for B}Given marginals (\ref{eq:state signaling B})
of a Bell-system, 
\begin{equation}
C_{0}=\left(\begin{array}{l}
\frac{1}{2}\left|\left\langle A_{11}\right\rangle -\left\langle A_{12}\right\rangle \right|,\frac{1}{2}\left|\left\langle A_{21}\right\rangle -\left\langle A_{22}\right\rangle \right|,\\
\frac{1}{2}\left|\left\langle B_{11}\right\rangle -\left\langle B_{21}\right\rangle \right|,\frac{1}{2}\left|\left\langle B_{12}\right\rangle -\left\langle B_{22}\right\rangle \right|
\end{array}\right).\label{eq:C_0 general B}
\end{equation}

\end{lem}
The proof of this and subsequent formal statements is relegated to
Appendix. Note that under no-signaling we have $C_{0}=\mathbf{0}$,
in accordance with (\ref{eq:identity connections B}). The question
we ask is whether this $C_{0}$ is compatible with the observed distributions
of the pairs $\left(A_{ij},B_{ij}\right)$. If it is, the Bell-system
exhibits no contextuality. If it is not, then contextuality is present,
and a measure of its degree is easily computed as shown below. 

The compatibility of $C_{0}$ with the observed pairs of random outputs
means that a joint distribution can be imposed on all eight random
variables so that it is consistent with both $C_{0}$ and the observed
pairs. In other words, each of the $2^{8}$ possible combinations
\begin{equation}
A_{11}=\pm1,B_{11}=\pm1,\ldots,A_{22}=\pm1,B_{22}=\pm1
\end{equation}
can be assigned a probability, so that the probabilities for all combinations
containing, say, $A_{12}=1$ and $B_{12}=-1$ sum to the observed
$\Pr\left[A_{12}=1,B_{12}=-1\right]$; and the probabilities for all
combinations containing unequal values of, say, $B_{12}$ and $B_{22}$
sum to $\Pr\left[B_{12}\not=B_{22}\right]$ in $C_{0}$. 
\begin{thm}[non-contextuality criterion for Bell-systems]
\label{thm:main B} A Bell-system exhibits no contextuality, i.e.,
$C_{0}$ in (\ref{eq:C_0 general B}) is compatible with the observed
pairs \textup{$\left(A_{ij},B_{ij}\right)_{i,j\in\left\{ 1,2\right\} }$}\textup{\emph{,
if and only if}}\emph{ }\textup{\emph{
\begin{equation}
\max_{i,j\in\left\{ 1,2\right\} }\left|\begin{array}{l}
\left\langle A_{11}B_{11}\right\rangle +\left\langle A_{12}B_{12}\right\rangle \\
+\left\langle A_{21}B_{21}\right\rangle +\left\langle A_{22}B_{22}\right\rangle -2\left\langle A_{ij}B_{ij}\right\rangle 
\end{array}\right|\leq2(1+\Delta_{0}),\label{eq:CHSH general}
\end{equation}
where $\Delta_{0}$ is the sum of the components of $C_{0}$,
\begin{equation}
\Delta_{0}=\frac{1}{2}\left(\begin{array}{l}
\left|\left\langle A_{11}\right\rangle -\left\langle A_{12}\right\rangle \right|+\left|\left\langle A_{21}\right\rangle -\left\langle A_{22}\right\rangle \right|\\
+\left|\left\langle B_{11}\right\rangle -\left\langle B_{21}\right\rangle \right|+\left|\left\langle B_{12}\right\rangle -\left\langle B_{22}\right\rangle \right|
\end{array}\right).\label{eq:Delta_0}
\end{equation}
}}
\end{thm}
For the LG-system the situation is analogous. We consider a vector
of probabilities 
\begin{equation}
C'=\left(\Pr\left[Q_{12}\not=Q_{13}\right],\Pr\left[Q_{21}\not=Q_{23}\right],\Pr\left[Q_{31}\not=Q_{32}\right]\right)
\end{equation}
and determine $C_{0}'$ with the minimum values of these probabilities
allowed by the system's marginals 
\begin{equation}
\left(\left\langle Q_{12}\right\rangle ,\left\langle Q_{13}\right\rangle ,\left\langle Q_{21}\right\rangle ,\left\langle Q_{23}\right\rangle ,\left\langle Q_{31}\right\rangle ,\left\langle Q_{32}\right\rangle \right).\label{eq:state signaling LG}
\end{equation}

\begin{lem}
\label{lem:C0 for LG}Given marginals (\ref{eq:state signaling LG})
of an LG-system, 
\begin{equation}
C'_{0}=\left(0,\frac{1}{2}\left|\left\langle Q_{21}\right\rangle -\left\langle Q_{23}\right\rangle \right|,\frac{1}{2}\left|\left\langle Q_{31}\right\rangle -\left\langle Q_{32}\right\rangle \right|\right).\label{eq:C_0 general LG}
\end{equation}

\end{lem}
Note that, by causality considerations, $\left|\left\langle Q_{12}\right\rangle -\left\langle Q_{13}\right\rangle \right|$
in $C'_{0}$ must equal zero (but it need not be in a generalized
treatment, if $t_{1},t_{2},t_{3}$ are treated as labels other than
time moments). 
\begin{thm}[non-contextuality criterion for LG-systems]
\label{thm:main LG} An LG-system exhibits no contextuality, i.e.,
\textup{$C'_{0}$} in (\ref{eq:C_0 general LG}) is compatible with
the observed pairs $\left(Q_{12},Q_{21}\right),\left(Q_{13},Q_{31}\right),\left(Q_{23},Q_{32}\right)$,
if and only if
\begin{equation}
\begin{array}{l}
-1-2\Delta'_{0}\le\left\langle Q_{12}Q_{21}\right\rangle +\left\langle Q_{13}Q_{31}\right\rangle +\left\langle Q_{23}Q_{32}\right\rangle \\
\leq1+2\Delta'_{0}+2\max\left\{ \left\langle Q_{12}Q_{21}\right\rangle ,\left\langle Q_{13}Q_{31}\right\rangle ,\left\langle Q_{23}Q_{32}\right\rangle \right\} ,
\end{array}\label{eq:LG general}
\end{equation}
where\emph{ $\Delta'_{0}$ }\textup{\emph{is the sum of the components
of $C'_{0}$,
\begin{equation}
\Delta_{0}^{'}=\frac{1}{2}\left(\left|\left\langle Q_{21}\right\rangle -\left\langle Q_{23}\right\rangle \right|+\left|\left\langle Q_{31}\right\rangle -\left\langle Q_{32}\right\rangle \right|\right).\label{eq:Delta_0 prime}
\end{equation}
}}
\end{thm}
Under no-signaling condition, $\Delta_{0}$ and $\Delta'_{0}$ are
zero, and Theorems \ref{thm:main B} and \ref{thm:main LG} reduce
to the traditional non-contextuality criteria (\ref{eq:marginal selectivity B})-(\ref{eq:CHSH})
and (\ref{eq:marginal selectivity LG})-(\ref{eq:LGSZ}), respectively.
Note also that a Bell-system with $\Delta_{0}>1$ and an LG-system
with $\Delta_{0}^{'}>1$ are necessarily non-contextual, as (\ref{eq:CHSH general})
and, respectively, (\ref{eq:LG general}) then cannot be violated.

\section{Degree of Contextuality Under Signaling}

A measure of contextuality is based on the same compatibility-under-constraints
considerations as the criteria just derived. For a Bell-system, let
$\Delta_{\min}$ be the minimum value of 
\begin{equation}
\Delta=\begin{array}{l}
\Pr\left[A_{11}\not=A_{12}\right]+\Pr\left[A_{21}\not=A_{22}\right]\\
+\Pr\left[B_{11}\not=B_{21}\right]+\Pr\left[B_{12}\not=B_{22}\right]
\end{array}\label{eq:Delta}
\end{equation}
that is compatible with the observed pairs $\left(A_{ij},B_{ij}\right)_{i,j\in\left\{ 1,2\right\} }$.
It follows from the previous that the system exhibits contextuality
if and only if this $\Delta_{\min}$ exceeds the value of $\Delta_{0}$
in (\ref{eq:Delta_0}). It is natural therefore to define the degree
of contextuality in a Bell system as 
\begin{equation}
\max\left(0,\Delta_{\min}-\Delta_{0}\right)
\end{equation}
This value is well-defined and given by
\begin{thm}[contextuality degree in Bell-systems]
\label{thm:-The-degree B} The degree of contextuality in a Bell-system
is\textup{
\begin{equation}
\max\left\{ \begin{array}{l}
0,\frac{1}{2}\max_{i,j\in\left\{ 1,2\right\} }\left|\begin{array}{l}
\left\langle A_{11}B_{11}\right\rangle +\left\langle A_{12}B_{12}\right\rangle \\
+\left\langle A_{21}B_{21}\right\rangle +\left\langle A_{22}B_{22}\right\rangle \\
-2\left\langle A_{ij}B_{ij}\right\rangle 
\end{array}\right|-1-\Delta_{0}\end{array}\right\} .\label{eq:contextuality B}
\end{equation}
}
\end{thm}
The degree of contextuality thus is always nonnegative. It equals
zero if and only if $\Delta_{\min}=\Delta_{0}$, which is equivalent
to (\ref{eq:CHSH general}). Returning to our motivating example (\ref{eq:PR+signaling}),
the degree of contextuality there is $\max\left(0,2\delta-1-2\left|\varepsilon\right|\right)$,
changing continuously with $\varepsilon$. 

For LG-systems the degree of contextuality is defined analogously,
as 
\[
\max\left(0,\Delta'_{\min}-\Delta'_{0}\right),
\]
where $\Delta'_{\min}$ is the smallest value of
\begin{equation}
\Delta'=\Pr\left[Q_{12}\not=Q_{13}\right]+\Pr\left[Q_{23}\not=Q_{21}\right]+\Pr\left[A_{32}\not=A_{31}\right]\label{eq:Delta'}
\end{equation}
compatible with the observed pairs $\left(Q_{ij},Q_{ji}\right)_{i<j\in\left\{ 1,2,3\right\} }$.
\begin{thm}[contextuality degree in LG-systems]
\label{thm:-The-degree LG} The degree of contextuality in an LG-system
is\textup{
\begin{equation}
\max\left\{ \begin{array}{r}
0,\frac{1}{2}\max\left\{ \begin{array}{l}
\pm\left\langle Q_{12}Q_{21}\right\rangle \pm\left\langle Q_{13}Q_{31}\right\rangle \\
\pm\left\langle Q_{23}Q_{32}\right\rangle :\\
\textnormal{number of minuses}\\
\textnormal{ is odd}
\end{array}\right\} -\frac{1}{2}-\Delta'_{0}\end{array}\right\} .\label{eq:contextuality LG}
\end{equation}
}
\end{thm}
\global\long\def\thesection{A}

\numberwithin{equation}{section}\numberwithin{thm}{section}\numberwithin{figure}{section}

\setcounter{equation}{0}\setcounter{thm}{0}

\section*{Appendix: Proofs}

We use the convenient notion of a (probabilistic) connection \cite{21DK2014FFOP,5DK2013LNCS},
as defined in Fig. 1. We also make use of two functions: for any natural
$r$, $s_{0}\left(x_{1},\ldots,x_{2r}\right)$ stands for $\max\left\{ \left(\pm x_{1}\ldots\pm x_{2r}\right):\#\textnormal{ of minuses is even}\right\} $,
and $s_{1}\left(x_{1},\ldots,x_{r}\right)$ denotes $\max\left\{ \left(\pm x_{1}\ldots\pm x_{r}\right):\textnormal{\# of minuses is odd}\right\} $.
\begin{proof}[Proof of Lemma \ref{lem:C0 for B}]
 Consider, e.g., the distribution of the connection $\left(A_{11},A_{12}\right)$:\begin{equation}%
\begin{tabular}{|c|cc|}
\cline{2-3} 
\multicolumn{1}{c|}{} & $A_{12}=+1$ & $A_{12}=-1$\tabularnewline
\hline 
$A_{11}=+1$ & $p$ & $\Pr\left[A_{11}=1\right]-p$\tabularnewline
$A_{11}=-1$ & $\Pr\left[A_{12}=1\right]-p$ & $\ldots$\tabularnewline
\hline 
\end{tabular}\end{equation}The largest possible value for $p$ is $\min\left\{ \Pr\left[A_{11}=1\right],\Pr\left[A_{12}=1\right]\right\} $,
whence the minimum of $\Pr\left[A_{11}\not=A_{12}\right]$, which
is the sum of the entries on the minor diagonal, is $\left|\Pr\left[A_{11}=1\right]-\Pr\left[A_{12}=1\right]\right|=\frac{1}{2}\left|\left\langle A_{11}\right\rangle -\left\langle A_{12}\right\rangle \right|$. 
\end{proof}
Lemma \ref{lem:C0 for LG} is proved in the same way. 

The theorems of this paper are based on the following four lemmas.
Their proofs are computer-assisted, as they boil down to symbolically
solving large systems of linear inequalities. 
\begin{lem}
\label{lem:1}The necessary and sufficient condition for the connections
$\left(\left(A_{i1},A_{i2}\right),\left(B_{1j},B_{2j}\right)\right)_{i,j\in\left\{ 1,2\right\} }$
to be compatible with the observed pairs $\left(A_{ij},B_{ij}\right)_{i,j\in\left\{ 1,2\right\} }$
is 
\begin{equation}
\begin{array}{l}
s_{0}\left(\left\langle A_{11}B_{11}\right\rangle ,\left\langle A_{12}B_{12}\right\rangle ,\left\langle A_{21}B_{21}\right\rangle ,\left\langle A_{22}B_{22}\right\rangle \right)\\
+s_{1}\left(\left\langle A_{11}A_{12}\right\rangle ,\left\langle B_{11}B_{21}\right\rangle ,\left\langle A_{21}A_{22}\right\rangle ,\left\langle B_{12}B_{22}\right\rangle \right)\le6,\\
s_{1}\left(\left\langle A_{11}B_{11}\right\rangle ,\left\langle A_{12}B_{12}\right\rangle ,\left\langle A_{21}B_{21}\right\rangle ,\left\langle A_{22}B_{22}\right\rangle \right)\\
+s_{0}\left(\left\langle A_{11}A_{12}\right\rangle ,\left\langle B_{11}B_{21}\right\rangle ,\left\langle A_{21}A_{22}\right\rangle ,\left\langle B_{12}B_{22}\right\rangle \right)\le6.
\end{array}\label{eq:compatibility}
\end{equation}
\end{lem}
\begin{proof}
The joint distribution of the eight random variables $A_{ij},B_{ij}$,
$i,j\in\left\{ 1,2\right\} $, is fully described by the vector $\mathbf{q}\in[0,1]^{n},$
$q_{1}+\dots+q_{n}=1$, consisting of the probabilities of the $n=2^{8}$
different combinations of the values of the $8$ random variables.
We define a vector $\mathbf{p}\in[0,1]^{m}$, $m=32$, consisting
of the $16$ observed probabilities $\Pr[A_{ij}=a,\ B_{ij}=b]$ and
the $16$ connection probabilities $\Pr[A_{i1}=a,\ A_{i2}=a']$ and
$\Pr[B_{1j}=b,\ B_{2j}=b']$, where $a,a',b,b'\in\{-1,1\}$ and $i,j\in\{1,2\}$.
The observed probabilities are compatible with the connection probabilities
if and only if there exists an $n$-vector $\mathbf{q}\ge0$ (componentwise)
such that $\mathbf{p}=M\mathbf{q}$, where $M\in\{0,1\}^{m\times n}$
determines which components of $\mathbf{q}$ sum to each component
of $\mathbf{p}$. As described in Text~S3 of Ref. \cite{22DK2013PLOS1},
the set of vectors $\mathbf{p}$ forms a polytope whose vertices are
given by the columns of $M$ and whose half-space representation can
be obtained by a facet enumeration algorithm. This half-space representation
consists of $160$ inequalities, as well as $16$ equations ensuring
that the marginals of the observed probabilities agree with those
of the connections and the probabilities are properly normalized.
Expressing the probabilities in $\mathbf{p}$ in terms of the observed
and connection expectations $\left(\left\langle A_{ij}B_{ij}\right\rangle ,\left\langle A_{ij}\right\rangle ,\left\langle B_{ij}\right\rangle ,\left\langle A_{i1}A_{i2}\right\rangle ,\left\langle B_{1j}B_{2j}\right\rangle \right)$,
$i,j\in\{1,2\}$, the $16$ equations become identically true (the
parameterization alone guarantees them), and of the $160$ inequalities,
$128$ turn into exactly those represented by \eqref{eq:compatibility};
the remaining $32$ inequalities need not be listed as they are constraints
of the form $-1+|\left\langle A\right\rangle +\left\langle B\right\rangle |\le\left\langle AB\right\rangle \le1-|\left\langle A\right\rangle -\left\langle B\right\rangle |$,
trivially following from the nonnegativity of probabilities.
\end{proof}
This proof is different from the similar result in Ref. \cite{22DK2013PLOS1}
in that the parameterization for the probabilities in $\mathbf{p}$
is more general (allowing for arbitrary marginals of the eight random
variables) and so we obtain a more general condition for the compatibility
of observed and connection probabilities.

\begin{figure}
\medskip{}
\fbox{\begin{minipage}[t]{1\columnwidth}%
\[
\begin{array}{c}
\xymatrix{A_{12}\ar@{.>}[d]\ar[r] & B_{12}\ar@{.>}[r]\ar[l] & B_{22}\ar@{.>}[l]\ar[r] & A_{22}\ar@{.>}[d]\ar[l]\\
A_{11}\ar@{.>}[u]\ar[r] & B_{11}\ar@{.>}[r]\ar[l] & B_{21}\ar@{.>}[l]\ar[r] & A_{21}\ar@{.>}[u]\ar[l]
}
\end{array}\tag{Bell-system}
\]
\[
\begin{array}{c}
\xymatrix{Q_{21}\ar@{.>}[dr]\ar[r] & Q_{12}\ar@{.>}[r]\ar[l] & Q_{13}\ar@{.>}[l]\ar[r] & Q_{31}\ar@{.>}[dl]\ar[l]\\
 & Q_{23}\ar@{.>}[ul]\ar[r] & Q_{32}\ar@{.>}[ur]\ar[l]
}
\end{array}\tag{LG-system}
\]
\end{minipage}}

\protect\caption[.]{Random variables involved in the Bell-system and LG-system. The pairs
of random variables whose joint distributions are empirically observed,
e.g., $\left(A_{12},B_{12}\right)$ and $\left(Q_{12},Q_{21}\right)$,
are indicated by solid double-arrows. The pairs of random variables
forming probabilistic connections (with unobservable joint distributions)
are indicated by point double-arrows, e.g., $\left(A_{11},A_{12}\right)$
and $\left(Q_{12},Q_{13}\right)$. Lemmas \ref{lem:C0 for B} and
\ref{lem:C0 for LG} are about connections whose components are as
close to being identical as possible; Theorems \ref{thm:main B} and
\ref{thm:main LG} are about connections compatible with the observed
pairs.\label{fig:A-representation-of} }
\end{figure}
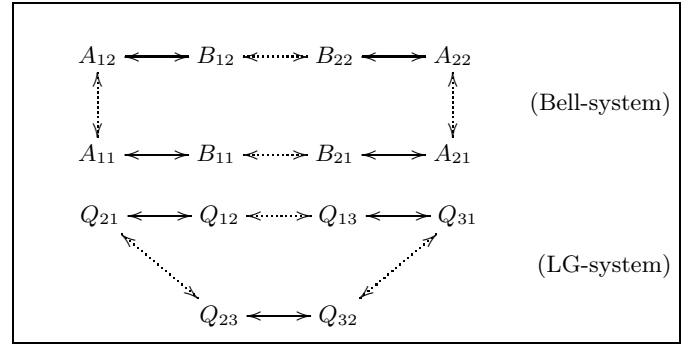

\begin{lem}
\label{lem:1-1}The necessary and sufficient condition for the connections
$\left(Q_{12},Q_{13}\right)$, $\left(Q_{21},Q_{23}\right)$, $\left(Q_{31},Q_{32}\right)$
to be compatible with the observed pairs $\left(Q_{12},Q_{21}\right)$,
$\left(Q_{13},Q_{31}\right)$, $\left(Q_{23},Q_{32}\right)$ is 
\begin{equation}
s_{1}\left(\begin{array}{r}
\left\langle Q_{12}Q_{21}\right\rangle ,\left\langle Q_{13}Q_{31}\right\rangle ,\left\langle Q_{23}Q_{32}\right\rangle ,\\
\left\langle Q_{12}Q_{13}\right\rangle ,\left\langle Q_{21}Q_{23}\right\rangle ,\left\langle Q_{31}Q_{32}\right\rangle 
\end{array}\right)\le4.\label{eq:compatibility-1}
\end{equation}

\end{lem}
The proof is analogous to that of Lemma~\ref{lem:1}.
\begin{lem}
\label{lem:2}If the connections $\left(\left(A_{i1},A_{i2}\right),\left(B_{1j},B_{2j}\right)\right)_{i,j\in\left\{ 1,2\right\} }$
are compatible with the observed pairs $\left(A_{ij}B_{ij}\right)_{i,j\in\left\{ 1,2\right\} }$,
then, with $\Delta$ defined as in (\ref{eq:Delta}), 
\begin{equation}
\begin{array}{l}
\Delta\ge-1+\frac{1}{2}s_{1}\left(\begin{array}{r}
\left\langle A_{11}B_{11}\right\rangle ,\left\langle A_{12}B_{12}\right\rangle ,\\
\left\langle A_{21}B_{21}\right\rangle ,\left\langle A_{22}B_{22}\right\rangle 
\end{array}\right),\\
\Delta\ge\frac{1}{2}\left(\begin{array}{l}
\left|\left\langle A_{11}\right\rangle -\left\langle A_{12}\right\rangle \right|+\left|\left\langle A_{21}\right\rangle -\left\langle A_{22}\right\rangle \right|\\
+\left|\left\langle B_{11}\right\rangle -\left\langle B_{21}\right\rangle \right|+\left|\left\langle B_{12}\right\rangle -\left\langle B_{22}\right\rangle \right|
\end{array}\right),\\
\Delta\le5-\frac{1}{2}s_{1}\left(\begin{array}{r}
\left\langle A_{11}B_{11}\right\rangle ,\left\langle A_{12}B_{12}\right\rangle ,\\
\left\langle A_{21}B_{21}\right\rangle ,\left\langle A_{22}B_{22}\right\rangle 
\end{array}\right),\\
\Delta\le4-\frac{1}{2}\left(\begin{array}{l}
\left|\left\langle A_{11}\right\rangle +\left\langle A_{12}\right\rangle \right|+\left|\left\langle A_{21}\right\rangle +\left\langle A_{22}\right\rangle \right|\\
+\left|\left\langle B_{11}\right\rangle +\left\langle B_{21}\right\rangle \right|+\left|\left\langle B_{12}\right\rangle +\left\langle B_{22}\right\rangle \right|
\end{array}\right).
\end{array}\label{eq:Delta_system B}
\end{equation}
Conversely, if these inequalities are satisfied for a given value
of $\Delta$, then the connection distributions can always be chosen
so that yield this value of $\Delta$ and are compatible with the
distributions of the observed pairs.\end{lem}
\begin{proof}
Given the 160 inequalities of Lemma~\ref{lem:1} (characterizing
the compatibility of the connections with the observed pairs), we
add to this linear system the equation defining $\Delta$ in terms
of the expectations $\left(\left\langle A_{i1}A_{i2}\right\rangle ,\left\langle B_{1j}B_{2j}\right\rangle ,\left\langle A_{ij}\right\rangle ,\left\langle B_{ij}\right\rangle \right)_{i,j\in\{1,2\}}$.
Then we use this equation to eliminate one of the connection expectation
variables $\left(\left\langle A_{i1}A_{i2}\right\rangle ,\left\langle B_{1j}B_{2j}\right\rangle \right)_{i,j\in\{1,2\}}$
from the system (by solving the equation for this variable and then
substituting the solution everywhere else). After that, we eliminate
the three remaining connection expectation variables one by one using
the Fourier-Motzkin elimination algorithm \cite{FME}. Then we remove
any redundant inequalities from the system by linear programming using
the algorithm described in Ref. \cite{22DK2013PLOS1}, Text~S3. After
having eliminated all connection expectation variables and having
deleted the inequalities following from the nonnegativity of probabilities,
we are left with the system \eqref{eq:Delta_system B}. The Fourier-Motzkin
elimination algorithm guarantees that the resulting system has a solution
precisely when the original system has a solution with some values
of the eliminated variables.\end{proof}
\begin{lem}
\label{lem:2-1}If the connections $\left(Q_{12},Q_{13}\right),\left(Q_{21},Q_{23}\right),\left(Q_{31},Q_{32}\right)$
are compatible with the observed pairs $\left(Q_{12},Q_{21}\right),\left(Q_{13},Q_{31}\right),\left(Q_{23},Q_{32}\right)$,
then, with $\Delta'$ defined as in (\eqref{eq:Delta'}), 
\begin{equation}
\begin{array}{l}
\Delta'\ge-\frac{1}{2}+\frac{1}{2}s_{1}\left(\left\langle Q_{12}Q_{21}\right\rangle ,\left\langle Q_{13}Q_{31}\right\rangle ,\left\langle Q_{23}Q_{32}\right\rangle \right),\\
\Delta'\ge\frac{1}{2}\left(\begin{array}{l}
\left|\left\langle Q_{12}\right\rangle -\left\langle Q_{13}\right\rangle \right|\\
+\left|\left\langle Q_{21}\right\rangle -\left\langle Q_{23}\right\rangle \right|+\left|\left\langle Q_{31}\right\rangle -\left\langle Q_{32}\right\rangle \right|
\end{array}\right),\\
\Delta'\le\frac{7}{2}-\frac{1}{2}s_{1}\left(\left\langle Q_{12}Q_{21}\right\rangle ,\left\langle Q_{13}Q_{31}\right\rangle ,\left\langle Q_{23}Q_{32}\right\rangle \right),\\
\Delta'\le3-\frac{1}{2}\left(\begin{array}{l}
\left|\left\langle Q_{12}\right\rangle +\left\langle Q_{13}\right\rangle \right|\\
+\left|\left\langle Q_{21}\right\rangle +\left\langle Q_{23}\right\rangle \right|+\left|\left\langle Q_{31}\right\rangle +\left\langle Q_{32}\right\rangle \right|
\end{array}\right).
\end{array}\label{eq:Delta_system LG}
\end{equation}
Conversely, if these inequalities are satisfied for a given value
of $\Delta'$, then the connection distributions can always be chosen
so that yield this value of $\Delta'$ and are compatible with the
distributions of the observed pairs.
\end{lem}
The proof is analogous to that of Lemma~\eqref{lem:2}.
\begin{proof}[Proof of Theorems \ref{thm:main B} and \ref{thm:-The-degree B}]
 Inequalities \eqref{eq:Delta_system B} in Lemma \ref{lem:1-1}
can be easily checked to be mutually compatible, whence $\Delta_{\min}$
is the larger of the two right-hand expressions in the first and third
of them. Note that $s_{1}(\cdots)$ is the same as $\max\left|\ldots\right|$-part
of \eqref{eq:contextuality B}. This proves Theorem \ref{thm:main B},
and Theorem \ref{thm:-The-degree B} follows as an explication of
$\Delta_{\min}=\Delta_{0}$. 
\end{proof}
The proofs of Theorems \ref{thm:main LG} and \ref{thm:-The-degree LG}
follows from Lemma \ref{lem:2-1} analogously. 

\smallskip{}

This work was supported by NSF grant SES-1155956. The authors are
grateful to J. Acacio de Barros, Gary Oas, Jan-Åke Larsson, and Guido
Bacciagaluppi for helpful discussions of issues related to probabilistic
contextuality.


\begin{thebibliography}{10}
\bibitem{17note1}Suppes \& Zanotti \cite{16Suppes} prove (\ref{eq:LGSZ})
for uniform marginals only, $\left\langle Q_{ij}\right\rangle =0$,
while Leggett \& Garg \cite{11Leggett} use only one part of the double-inequality.
The generality of (\ref{eq:LGSZ}) follows as a special case from
our Theorem \ref{thm:main LG}.
\end{thebibliography}
\end{document}